\documentclass[journal]{IEEEtran}
\usepackage{latexsym}
\usepackage{graphicx}
\usepackage{array}
\usepackage{amsmath}
\usepackage{amsfonts}
\usepackage{amssymb}
\usepackage{amsthm}

\newtheorem{thm}{Theorem}

\newtheorem{prop}{Proposition}

\newtheorem{cor}{Corollary}

\newcommand{\bx} {\boldsymbol{x}}

\newcommand{\by} {\boldsymbol{y}}

\newcommand{\bH} {\boldsymbol{H}}

\newcommand{\bA} {\boldsymbol{A}}
\newcommand{\bD} {\boldsymbol{D}}
\newcommand{\bI} {\boldsymbol{I}}
\newcommand{\bR} {\boldsymbol{R}}

\newcommand{\bu} {\boldsymbol{u}}
\newcommand{\bh} {\boldsymbol{h}}
\newcommand{\bW} {\boldsymbol{W}}

\newcommand{\bM} {\boldsymbol{M}}

\newcommand{\bLam} {\boldsymbol{\Lambda}}
\newcommand{\gl}{\lambda}

\newcommand{\bxi} {\boldsymbol{\xi}}

\def\bal#1\eal{\begin{align}#1\end{align}}
\newcommand{\bp} {\begin{proof}}
\newcommand{\ep} {\end{proof}}

\newcommand{{\bRF}} {\right\}}

\DeclareMathOperator{\tr}{tr}
\DeclareMathOperator{\diag}{diag}

\begin{document}

\title{On The Capacity of Gaussian MIMO Channels Under The Joint Power Constraints}

\author{Sergey Loyka

\vspace*{-1\baselineskip}

\thanks{S. Loyka is with the School of Electrical Engineering and Computer Science, University of Ottawa, Ontario, Canada, e-mail: sergey.loyka@ieee.org}

}

\maketitle


\begin{abstract}
The capacity and optimal signaling over a fixed Gaussian MIMO channel are considered under the joint total and per-antenna power constraints (TPC and PAC). While the general case remains an open problem, a closed-form full-rank solution is obtained along with its sufficient and necessary conditions. The conditions for each constraint to be inactive are established. The high and low-SNR regimes are studied. Isotropic signaling is shown to be optimal in the former case while rank-1 signaling (beamforming) is not necessarily optimal in the latter case. Unusual properties of optimal covariance under the joint constraints are pointed out.
\end{abstract}

\vspace*{-1\baselineskip}
\section{Introduction}

Multi-antenna (MIMO) systems have been widely accepted by both academia and industry due to their high spectral efficiency \cite{Tse-05}. Massive MIMO is considered as a key technology for future 5G systems to meet ever-increasing traffic demand over a limited bandwidth available \cite{Shafi-17}. The capacity of a fixed Gaussian MIMO channel and its optimal signaling strategy are well-known under the total transmit power constaint (TPC): its is on the channel eigenmodes with power allocation given by the water-filling (WF) procedure \cite{Tsybakov-65}\cite{Telatar-95}. While the TPC is motivated by a limited power (energy) supply, individual per-antenna powers can also be limited when each antenna is equipped with its own amplifier (of limited power), in either collocated or distributed implementations, hence motivating per-antenna power constraint (PAC), as in \cite{Yu-07}-\cite{Tuninetti-14}. The capacity of a fixed Gaussian MISO channel under the PAC has been established in \cite{Vu-11}, which is significantly different from the standard WF solution and is equivalent to the equal-gain transmission (EGT) with phases adjusted to compensate for the channel phase shifts. This problems remains open in the general MIMO case while a numerical algorithm was proposed in \cite{Vu-11b} and a closed-form full-rank solution was obtained in \cite{Tuninetti-14}. Single-user PAC-constrained results were extended to multiple access channel in \cite{Zhu-12} via a numerical optimization algorithm. The capacity of ergodic fading MISO channel under long-term average PAC and full CSI at both ends was established in \cite{Maamari-14}.

One may also consider the joint constraints, i.e. the TPC and the PAC simultaneously. This is motivated by the scenario with limited overall power budget and where each antenna is equipped with its own power amplifier. The capacity of fixed Gaussian MISO channel under the joint TPC and PAC has been established in \cite{Loyka-16}, where it was shown that the optimal signaling is a combination of EGT and maximum ratio transmission (MRT), with phase shifts adjusted to compensate channel-induced phase shifts. Following the remark in \cite{Tuninetti-14}, the MISO result can be also adapted to any rank-1 MIMO channel. This result was further extended to fading MIMO channels in \cite{Loyka-17}, where it was shown that isotropic signaling is optimal if the fading distribution is right-unitary-invariant. A sub-optimal signaling strategy for ergodic-fading MIMO channel was developed in \cite{Khoshnevisan-12} under the long-term TPC and short-term PAC. The general MIMO case under the joint power constraints remains an open problem. The key difficulty is the fact that, unlike the TPC only case, the feasible set of transmit covariance matrices is not isotropic anymore (due to the PAC) and hence the tools developed under the TPC (which exploit this symmetry) cannot be used anymore. New tools are needed.

The present paper partially closes this gap by obtaining a closed-form full-rank solution for the optimal signaling in a fixed Gaussian MIMO channel under the joint constraints, thus extending earlier results in  \cite{Tuninetti-14} and \cite{Loyka-16}. Sufficient and necessary conditions for optimal signaling to be of full rank are also established. Optimal signaling under the joint constraints is shown to have properties significantly different from those under the TPC only. Namely, (i) optimal covariance is not necessarily unique (multiple solutions are possible), (ii) it can be of full-rank even when the channel is rank-deficient, (iii) signaling on the channel eigenmodes is not optimal anymore (unless all PACs are inactive). It is the inter-play between the TPC and PAC that induces these unusual properties. The conditions when either TPC and PAC are inactive are given. The high and low-SNR regimes are studied. Isotropic signaling is shown to be optimal under the joint constraints in the former case while rank-1 signaling (beamforming) is not necessarily optimal in the latter case (in contrast to the standard WF signaling).

\textit{Notations}: bold lower-case letters denote column vectors while bold capital denote matrices; $\bR^+$ is Hermitian conjugation of $\bR$; $r_{ii}$ denotes $i$-th diagonal entry of $\bR$; $(\bR)_{ij}$ is $ij$-th entry of $\bR$, $\gl_i(\bR)$ is $i$-th eigenvalue of $\bR$, unless indicated otherwise, eigenvalues are in decreasing order: $\gl_1\ge \gl_2 \ge...$; $\bR \ge 0$ means that $\bR$ is positive semi-definite; $|\bR|$ is the determinant of $\bR$, $\bI$ is identity matrix of appropriate size, $(x)_+=\max(0,x)$.

\section{Channel Model and Capacity}
\label{sec.Channel Model}

Discrete-time model of a fixed Gaussian MISO channel can be put into the following form:
\bal
\label{eq.ch.model}
\by = \bH\bx +\bxi
\eal
where $\by, \bx, \bxi$ and $\bH$ are the received and transmitted signals, noise and channel respectively; $m$ is the number of transmit antennas. The noise is assumed to be complex Gaussian circularly-symmetric with zero mean and unit variance, so that power is also the SNR. Complex-valued channel model is assumed throughout the paper, with full channel state information available both at the transmitter and the receiver. The channel $\bH$ is fixed. Gaussian signaling is known to be optimal in this setting \cite{Tse-05}-\cite{Telatar-95} so that finding the channel capacity $C$ amounts to finding an optimal transmit covariance matrix $\bR$:
\bal
\label{eq.C.def}
C = \max_{\bR \in S_R} \ln|\bI+ \bW\bR|
\eal
where $\bW=\bH^+\bH$, $S_R$ is the constraint set. In the case of the TP constraint only, it takes the form
\bal
S_R=\{\bR: \bR\ge 0,\ \tr\bR \le P_T\},
\eal
where $P_T$ is the maximum total Tx power, and the optimal covariance is well-known: the optimal signaling is on the channel eigenmodes with optimal power allocation via the water-filling, which can be compactly expressed as
\bal
\bR^*_{WF}=(\gl^{-1}\bI-\bW^{-1})_+
\eal
where $(\bA)_+$ retains positive eigenmodes of Hermitian matrix $\bA$,
\bal
(\bA)_+ = \sum_{i: \gl_i(\bA)>0} \gl_i(\bA) \bu_i\bu_i^+
\eal
where $\bu_i$ is $i$-th eigenvector of $\bA$; $\gl>0$ is determined from the TPC $\tr\bR^*_{WF} = P_T$.

Under the PA constraints,
\bal
S_R=\{\bR: \bR\ge 0, r_{ii} \le P\},
\eal
where $r_{ii}$ is $i$-th diagonal entry of $\bR$ (the Tx power of $i$-th antenna), $P$ is the PA power constraint. No closed-form solution is known for the optimal covariance in the general case under this constraint, while such solutions are available in the MISO case \cite{Vu-11} and in the MIMO case when the optimal covariance is of full-rank \cite{Tuninetti-14}.

The joint power constraints, i.e. TPC and PAC, are motivated by practical designs where each antenna has its own amplifier (and hence PAC) while limited total power/energy supply motivates TPC. The optimal signaling and capacity have been found under the joint constraints for the MISO channel in \cite{Loyka-16}, while the general MIMO case remains an open problem.

The next section provides a closed-form full-rank solution for the MIMO case as well as sufficient and necessary conditions for this solution to hold and some related properties.

\section{The MIMO Capacity Under the Joint Constraints}
\label{sec.C.PA.TP}

Following the standard arguments, see e.g. \cite{Tse-05}-\cite{Telatar-95}, Gaussian signaling is still optimal under the joint constraints and the channel capacity $C$ is as in \eqref{eq.C.def}, where the constraint set $S_R$ is as follows:
\bal
\label{eq.SR.PA.TP}
S_R=\{\bR: \bR\ge 0, tr\bR \le P_T, r_{ii}\le P\}
\eal
and $P_T, P$ are the total and per-antenna constraint powers. Unfortunately, no closed-form solution is known for the optimal covariance in \eqref{eq.C.def} under the constraints in \eqref{eq.SR.PA.TP} in the general case. The following Theorem partially closes this gap and gives a closed-form full-rank solution for optimal signaling in this setting.

\begin{thm}
\label{thm.C}
Let the channel in \eqref{eq.ch.model} be of full column rank, $\bW=\bH^+\bH >0$, and let the per-antenna and total transmit constraint powers be sufficiently high,
\bal
\label{eq.thm.1}
P > \gl_m^{-1}(\bW),\ P_T > m\gl_m^{-1}(\bW) - \tr\bW^{-1}
\eal
Then, the optimal Tx covariance $\bR^*$ in \eqref{eq.C.def} under the TPC and PAC in \eqref{eq.SR.PA.TP} is of full-rank and is given by
\bal
\label{eq.R*.1}
\bR^* &= \min(P\bI, \gl^{-1}\bI - \bD(\bW^{-1})) - \bar{\bD}(\bW^{-1})\\
\label{eq.R*.2}
    &= \gl^{-1}\bI - \bW^{-1} - ((\gl^{-1}-P)\bI - \bD(\bW^{-1}))_+
\eal
where $\bD(\bW)$ retains only diagonal entries of $\bW$ (with all off-diagonal entries set to zero), and $\bar{\bD}(\bW)=\bW-\bD(\bW)$ retains off-diagonal entries only (with all diagonal entries set to zero), the operator $\min$ applies entry-wise, $\gl \ge 0$ is the Lagrange multiplier responsible for the total power constraint; $\gl=0$ if $mP \le P_T$; otherwise, it is determined as a unique solution of the following equation
\bal
\label{eq.R*.3}
\sum_i \min(P, \gl^{-1} - (\bW^{-1})_{ii}) = P_T
\eal
The capacity can be expressed as
\bal
C = \ln|\bW| + \sum_i \ln \min(\gl^{-1}, P+(\bW^{-1})_{ii})
\eal
\end{thm}
\begin{proof}
See Appendix.
\end{proof}

It follows from the proof that inactive TPC ($\gl=0$) implies that all PACs are active ($\gl_i>0$, $r_{ii}=P$ for all $i$). Hence, a single inactive PAC ($\gl_i=0$, $r_{ii}< P$ for some $i$) implies that the TPC is active ($\gl>0,\ \tr\bR^*=P_T$).

The expression in \eqref{eq.R*.2} has the following interpretation: its first part $\gl^{-1}\bI - \bW^{-1}$ is the standard full-rank WF solution under the TPC only, and its 2nd part $(\gl^{-1}\bI - \bD(\bW^{-1})- P\bI)_+$ is a correction term accounting for the PAC.

It follows from \eqref{eq.R*.1} that per-antenna powers are as follows:
\bal
r_{ii} = \min(P, \gl^{-1} - (\bW^{-1})_{ii}) >0
\eal
which also has an insightful interpretation: these powers are the minimum of those under the PAC and TPC individually (1st and 2nd term in the min operator, respectively).

Next, we show that the solution in Theorem \ref{thm.C} reduces to known solutions in some special cases.

\begin{cor}
In Theorem \ref{thm.C}, if the per-antenna constraint power $P$ is sufficiently high,
\bal
\label{eq.PA.inact}
P \ge \gl^{-1} - (\bW^{-1})_{ii}\ \forall i
\eal
then all PACs are inactive and \eqref{eq.R*.2} reduces to the standard WF solution,
\bal
\bR^* = \gl^{-1}\bI - \bW^{-1}
\eal
where $\gl^{-1} = m^{-1}(P_T+ \tr\bW^{-1})$.
\end{cor}

\begin{cor}
In Theorem \ref{thm.C}, if the TPC power $P_T$ is sufficiently high, $P_T \ge m P$, the TPC is inactive and \eqref{eq.R*.1} reduces to the PAC-only full-rank solution in \cite{Tuninetti-14},
\bal
\bR^* = P\bI - \bar{\bD}(\bW^{-1})
\eal
\end{cor}

\begin{cor}
In Theorem \ref{thm.C}, $i$-th PAC is active if and only if
\bal
\label{eq.PA.act}
(\bW^{-1})_{ii}< \gl^{-1} -P
\eal
\end{cor}

The Lagrange multiplier $\gl$ in Theorem \ref{thm.C} is determined from \eqref{eq.R*.3} when the TPC is active (otherwise, $\gl=0$). Since its left-hand side is a monotonically-decreasing function of $\gl$, bisection algorithm is an efficient (exponentially-fast) tool \cite{Boyd-04} to solve it. However, it requires lower and upper bounds of the solution for initialization. These are given below.

\begin{prop}
The Lagrange multiplier $\gl$ in Theorem \ref{thm.C} is bounded as follows:
\bal
0 \le \gl \le \gl_m(\bW)
\eal
\end{prop}
\begin{proof}
The lower bound is from dual feasibility. The upper bound is from \eqref{eq.NC.1}.
\end{proof}

It should be pointed out that while \eqref{eq.thm.1} are sufficient for the optimal signaling to be of full-rank, they are not necessary, i.e. there are cases where the optimal signaling is of full-rank even when \eqref{eq.thm.1} do not hold. The following proposition gives necessary conditions for an optimal covariance to be of full-rank.

\begin{prop}
Let $\bW>0$. The necessary conditions for optimal covariance $\bR^*$ to be of full rank are as follows:
\bal
\label{eq.NC.1}
P > \gl_1(\bar{\bD}(\bW^{-1})),\ \gl < \gl_m(\bW)
\eal
where 1st condition is also sufficient if $mP\le P_T$ (inactive TPC), and $\gl$ is determined from the TPC as in Theorem \ref{thm.C}.
\end{prop}
\begin{proof}
Using \eqref{eq.R*.1},
\bal\notag
\bR^* &= \min(P\bI, \gl^{-1}\bI - \bD(\bW^{-1})) - \bar{\bD}(\bW^{-1})\\
    &\le P\bI - \bar{\bD}(\bW^{-1})
\eal
so that $\bR^*>0$ implies $P\bI > \bar{\bD}(\bW^{-1})$ and hence 1st condition in \eqref{eq.NC.1}.  2nd condition is obtained from
\bal
0< \bR^* \le \gl^{-1}\bI - \bW^{-1}
\eal
\end{proof}

Based on this, the following procedure can be used to establish whether optimal covariance is of full-rank in general:

1. If $P_T \ge mP$, then 1st condition in \eqref{eq.NC.1} is both sufficient and necessary for $\bR^*>0$ and \eqref{eq.R*.1} applies.

2. If $P_T < mP$, define $\bR^*(\gl)$ for a given $\gl>0$ from \eqref{eq.R*.1} and find $\gl$ from \eqref{eq.R*.3}. If $\bR(\gl)^*>0$, then it is a solution; otherwise, optimal covariance is rank-deficient.

This procedure gives an exhaustive characterization of all cases when $\bR^*$ is of full rank for a full-rank channel (since KKT conditions are necessary for optimality).

In the following, we characterize the conditions when some constraints are inactive for a full-rank channel.

\begin{prop}
Let $\bW>0$. If the TPC is inactive, then all PACs are active. Hence, (i) when at least one PAC is inactive, the TPC is active; (ii) the TPC is inactive if and only if
\bal
mP\le P_T
\eal
\end{prop}
\begin{proof}
Follows from the stationarity condition in \eqref{eq.KKT.1}.
\end{proof}

It should be noted that this Proposition does not hold if the channel is rank-deficient, as the example below demonstrates.

\section{Orthogonal Channel}

To get additional insights, let us consider the case where the channel is orthogonal,
\bal
\bW=\bH^+\bH = \bD_w = \diag\{w_i\}
\eal
where $w_i=|\bh_i|^2$, and the columns of $\bH$ are orthogonal to each other: $\bh_i^+\bh_j=0,\ i\neq j$. One of the motivations of this is massive MIMO (a key technology for 5G \cite{Shafi-17}) where this orthogonality holds approximately for a large number of Rx antenna, with improving accuracy as the number of antennas increases \cite{Marzetta-16}.

In this case, $\bR^*$ is also diagonal and, from Theorem \ref{thm.C},
\bal
\bR^* = \min(P\bI,\gl^{-1}\bI-\bD_w^{-1})= \min(\bR^*_{PAC}, \bR^*_{TPC})
\eal
where $\bR^*_{PAC}=P\bI, \bR^*_{TPC}= \gl^{-1}\bI-\bD_w^{-1}$ are the optimal covariance matrices under the PAC and TPC only ($\gl$ is determined as in Theorem \ref{thm.C}).

Note that, for the orthogonal channel, isotropic signaling is optimal under the PAC but not under the TPC. Hence, we conclude that, in the massive MIMO setting, isotropic signaling is optimal under the PAC and suboptimal under the TPC or joint constraints (unless the TPC is inactive).

\section{High-SNR Regime}

It is well-known that isotropic signaling is optimal at high SNR for the standard WF solution (under the TPC only) in a full-rank channel,
\bal
\label{eq.WF.high.SNR}
\bR^*_{WF} \approx \frac{P_T}{m}\bI
\eal
when $P_T \gg m\gl_m^{-1}(\bW)$.
In this section, we establish the optimality of isotropic signaling under the joint constraints. As a first step, the following proposition shows that isotropic signaling is optimal at high SNR under the PAC.

\begin{prop}
\label{prop.PAC.iso}
Consider a full-rank channel ($\bW>0$). Isotropic signaling is optimal in this channel under the PAC in the high-SNR regime, i.e. when $P \gg \gl_m^{-1}(\bW)$,
\bal
\label{eq.PAC.high.SNR}
\bR^*_{PAC} \approx P\bI,\ C_{PAC}\approx \ln|\bW|+ m\ln P
\eal
\end{prop}
\begin{proof}
First, observe that
\bal
C_{PAC} \ge C(P\bI)
\eal
where $C(\bR) = \ln|\bI+\bW\bR|$, since $\bR=P\bI$ is feasible under the PAC. Next,
\bal
C_{PAC} \le C_{TPC}(mP)
\eal
where $C_{TPC}(mP)$ is the capacity under the TPC with the total power $P_T=mP$, since any feasible $\bR$ under the PAC, $r_{ii}\le P$, is also feasible under the TPC with $\tr\bR \le mP$. Using \eqref{eq.WF.high.SNR}, one obtains at high SNR $C_{TPC}(mP)\approx C(P\bI)$ and hence $C_{PAC}\approx C(P\bI)$ and \eqref{eq.PAC.high.SNR} follow.
\end{proof}

We are now in a position to establish the optimality of isotropic signaling under the joint (TPC + PAC) constraints at high SNR.

\begin{prop}
\label{prop.TPC+PAC.iso}
Consider a full-rank channel ($\bW>0$). Let $P^*=\min(P,P_T/m)$. Isotropic signaling is optimal in this channel under the joint constraints (TPC + PAC) in the high-SNR regime, i.e. when $P^* \gg \gl_m^{-1}(\bW)$,
\bal
\label{eq.TPC+PAC.high.SNR}
\bR^* \approx P^*\bI,\ C\approx \ln|\bW|+ m\ln P^*
\eal
\end{prop}
\begin{proof}
First, observe that
\bal
C \ge C(P^*\bI)
\eal
since $\bR=P^*\bI$ is feasible under the joint constraints: $\tr\bR \le P_T$ and $r_{ii}\le P$. Next,
\bal
C \le \min(C_{TPC}, C_{PAC})
\eal
and, at high SNR, $C_{TPC}\approx C(P_T\bI/m)$, $C_{PAC}\approx C(P\bI)$, and hence $C\approx C(P^*\bI)$, as desired. The inequality $P^* \gg \gl_m^{-1}(\bW)$ comes from the approximation $\ln(1+x)\approx \ln x$, which holds if $x\gg 1$.
\end{proof}

It is remarkable that, for any of the constraints considered here, isotropic signaling is optimal at high SNR. This simplifies the system design significantly as no feedback and no elaborate precoding are necessary for this signaling strategy. This also complements the respective result in \cite{Loyka-17} obtained for the right-unitary-invariant fading channel.

\section{Low-SNR Regime}

In this section, we consider the behaviour of optimal covariance in the low-SNR regime, namely, when
\bal
\label{eq.low.SNR}
\min(mP,P_T) \ll \gl_1^{-1}(\bW)
\eal

It is well-known that, for the standard WF solution (under the TPC only), the optimal signaling is beamforming (rank-1) at low SNR, $\bR^*_{WF} \approx P_T\bu_1\bu_1^+$, where $\bu_1$ is the eigenvector of $\bW$ corresponding to its largest eigenvalue. As the following example shows, this does not necessarily hold under the joint constraints.

\textit{Example}: Let $P_T=1.5\cdot 10^{-2}$, $P=10^{-2}$, and $\bW = \diag\{2,1\}$. It is straightforward to see that the optimal covariance is $\bR^*=10^{-2}\cdot \diag\{1,0.5\}$ in this case, i.e. full-rank and beamforming is not optimal, does not matter how low the SNR is. If, however, the per-antenna constraint power is increased to $P=1.5\cdot 10^{-2}$, all PACs become inactive and beamforming is optimal: $\bR^*=10^{-2}\cdot \diag\{1.5,0\}$.

Hence, we conclude that it is the interplay between the TPC and the PAC that makes a significant difference at low SNR while having negligible impact at high SNR: while the optimal signaling under the TPC, the PAC and the joint constraints are all isotropic at high SNR, they are quite different at low SNR.

\section{Properties of Optimal Covariance Under the Joint Constraints}

As it was seen in the previous sections, optimal signaling under the joint constraints can be significantly different from that under the TPC only. In the following, we point out additional significant differences.

1. The TPC can be inactive (unless $C=0$ - a trivial case not considered here).

2. Optimal covariance is not necessarily unique.

3. Optimal covariance can be of full-rank even when the channel is not.

4. Optimal signaling is not on the eigenmodes of $\bW$, unless it is diagonal or all PACs are inactive, and the capacity depends on its eigenvectors.

These unusual properties should be contrasted with those under the TPC only, where (i) the TPC is always active, (ii) the optimal covariance is always unique, (iii) optimal covariance is rank-deficient in a rank-deficient channel, and (iv) optimal signaling is on the channel eigenmodes and the capacity is independent of channel eigenvectors.

While Property 4 follows from Theorem \ref{thm.C}, the following example illustrates Properties 1-3.

\textit{Example}: Let $\bW = \diag\{1,0\},\ P_T=2,\ P=1$. It is straightforward to see that
\bal
\bR^*=\diag\{1,a\},\ 0 \le a \le 1
\eal
so that (i) $\bR^*$ is not unique, (ii) it is of full-rank when $a>0$, even though the channel is not, and (iii) the TPC is inactive if $a<1$. Note however that if the channel is enhanced to a full-rank one, $\bW = \diag\{1, b\},\ b>0$, then $\bR^*=\bI$ and all unusual properties disappear.

\section{An Example}

To further illustrate the MIMO channel capacity under the joint constraints, let us consider the following  channel:
\bal
\label{eq.ex}
\bW = \left(
        \begin{array}{cc}
          1 & 0.1 \\
          0.1 & 0.2 \\
        \end{array}
      \right)
\eal
with $P=3$. Fig. 1 shows the capacities under the TPC, the PAC and the joint constraints, in addition to that under isotropic signaling. It is clear that isotropic signaling is optimal at high SNR ($P_T\ge 6$), as expected from Propositions \ref{prop.PAC.iso} and \ref{prop.TPC+PAC.iso}. However, the low-SNR behaviour is different. Note that the TPC is inactive when $P_T\ge 6$ and the PAC is inactive when $P_T\le 3$ while both constraints are active in-between. Isotropic signaling is clearly sub-optimal when $P_T < 6$. When $P_T > 6$, the capacity saturates since the PAC dominates the performance, so that increasing total power supply beyond  $P_T = 6$ makes no difference.

\begin{figure}[t]
\centerline{\includegraphics[width=3.3in]{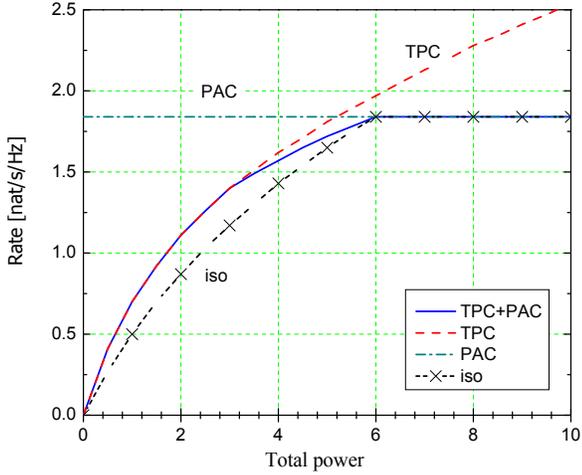}}
\caption{The capacity of the fixed MIMO channel in \eqref{eq.ex} under the TPC, the IPC and the joint constraints; isotropic signaling (iso) is also shown. $P=3$.}
\label{fig.1}
\end{figure}



\section{Appendix: Proof of Theorem \ref{thm.C}}

Since the problem in \eqref{eq.C.def} is convex and Slater's condition holds (as long as $P,\ P_T >0$), its KKT conditions are both sufficient and necessary for optimality \cite{Boyd-04}. The KKT conditions for this problem are as follows:
\bal
\label{eq.KKT.1}
- &(\bI+\bW\bR)^{-1}\bW -\bM +\gl\bI +\bLam = 0\\
\label{eq.KKT.2}
& \bM\bR=0,\ \gl(\tr\bR-P_T)=0,\ \gl_i(r_{ii}-P)=0\\
\label{eq.KKT.3}
& \tr\bR \le P_T,\ r_{ii}\le P, \ \bR \ge 0\\
\label{eq.KKT.4}
& \bM \ge 0,\ \gl \ge 0,\ \gl_i \ge 0
\eal
where $\gl,\ \gl_i$ are Lagrange multipliers (dual variables) responsible for the TPC and PAC, $\bM$ is the (matrix) Lagrange multiplier responsible for $\bR \ge 0$, $\bLam=\diag\{\gl_i\}$; \eqref{eq.KKT.1} is the stationarity condition, \eqref{eq.KKT.2} are complementary slackness conditions, \eqref{eq.KKT.3} and \eqref{eq.KKT.4} are primal and dual feasibility conditions. The key difficulty in solving analytically these conditions is that they are a system of non-linear matrix equalities and inequalities, and the PACs make the feasible set $S_R$ non-isotropic so that standard tools (e.g. Hadamard inequality) cannot be used. However, when $\bR$ is of full rank, the stationarity condition simplifies to
\bal
\label{eq.KKT.1a}
(\bR+\bW^{-1})^{-1} = \gl\bI +\bLam
\eal
since $\bM=0$ (from $\bM\bR=0$), so that
\bal
\label{eq.KKT.1b}
\bR = (\gl\bI +\bLam)^{-1} - \bW^{-1}
\eal
where $\bLam$ is determined from the PACs
\bal
\label{eq.KKT.1c}
r_{ii} = (\gl +\gl_i )^{-1} - (\bW^{-1})_{ii} \le P
\eal
and complementary slackness $\gl_i(r_{ii}-P)=0$ so that $\gl_i >0$ (active PAC) implies $r_{ii}= P$ and hence
\bal
\gl_i = (P+ (\bW^{-1})_{ii})^{-1} - \gl >0
\eal
Combining this with the case of inactive PAC $\gl_i=0$, one obtains
\bal
\label{eq.KKT.1c2}
r_{ii} = \gl^{-1} - (\bW^{-1})_{ii} \le P
\eal
and hence
\bal
\label{eq.KKT.1c3}
\gl_i = ((P+ (\bW^{-1})_{ii})^{-1} - \gl)_+ \ge 0
\eal
where $(x)_+=\max(x,0)$. It follows from \eqref{eq.KKT.1a} that off-diagonal parts of $\bR$ and $\bW^{-1}$ are the opposite of each other:
\bal
\bar{\bD}(\bR) = - \bar{\bD}(\bW^{-1})
\eal
and, from \eqref{eq.KKT.1c}-\eqref{eq.KKT.1c3}, that
\bal
\label{eq.KKT.1c3}
r_{ii} = \min(P, \gl^{-1} - (\bW^{-1})_{ii}) >0
\eal
from which \eqref{eq.R*.1} follows. \eqref{eq.R*.2} is a straightforward manipulation of \eqref{eq.R*.1}. \eqref{eq.R*.3} follows from the TPC, from which $\gl$ is found.

It remains to show that $\bR^* >0$. To this end, observe the following:
\bal\notag
\bR^* &= \min(P\bI, \gl^{-1}\bI - \bD(\bW^{-1})) - \bar{\bD}(\bW^{-1})\\ \notag
    &> \min(\gl_{w}^{-1}\bI, \gl_{w}^{-1}\bI - \bD(\bW^{-1})) - \bar{\bD}(\bW^{-1})\\
    &= \gl_{w}^{-1}\bI - \bW^{-1} \ge 0
\eal
where $\gl_{w}= \gl_m(\bW)$. The last inequality follows from $\gl_w\bI \le \bW$, while 1st inequality follows from $P> \gl_w^{-1}$ and $\gl^{-1} > \gl_w^{-1}$. To show the latter inequality, use \eqref{eq.R*.3} and note that
\bal\notag
f(\gl) = \sum_i \min(P, \gl^{-1} - (\bW^{-1})_{ii}) \le m\gl^{-1} -\tr\bW^{-1}
\eal
where $f(\gl)$ is a decreasing function (strictly so if at least one PAC is inactive), so that
\bal
f(\gl_w) \le m\gl_w^{-1} -\tr\bW^{-1} < P_T
\eal
from which it follows that $\gl< \gl_w$, as desired. This also implies that $r_{ii} > 0$ in \eqref{eq.KKT.1c3}, since
\bal
\gl^{-1} - (\bW^{-1})_{ii} > \gl_w^{-1} - (\bW^{-1})_{ii} \ge \gl_w^{-1} - \gl_w^{-1} = 0
\eal
where 2nd inequality is due to $(\bW^{-1})_{ii} \le \gl_1(\bW^{-1}) = \gl_w^{-1}$.

\end{document}